\documentclass[journal]{IEEEtran}
\usepackage{amsmath,latexsym,amssymb,bbm,amsthm, mathtools, verbatim, array}
\usepackage{color}

\newtheorem{theorem}{Theorem}

\newtheorem{lemma}{Lemma}
\newtheorem{definition}{Definition}

\def\reff#1{(\ref{#1})}

\def\sH{\mathcal{H}}

\newcommand*{\E}{\cal{E}}

\newcommand{\tomega}{{{\omega}}}
\newcommand{\supp}{\rm{supp}}

\newcommand*{\eps}{\varepsilon}
\newcommand*{\id}{I}

\newcommand*{\ket}[1]{| #1 \rangle}

\newcommand{\be}{\begin{equation}}
\newcommand{\bea}{\begin{eqnarray}}
\newcommand{\eea}{\end{eqnarray}}

\newcommand{\tlambda}{{{\Lambda}}}

\newcommand{\tr}{\mathrm{Tr}}

\newcommand{\ee}{\end{equation}}

\newlength{\blank}
\newcommand{\nc}{\newcommand}
\nc{\pw}{\mt{PW}}
\nc{\arbclass}{\mt{\Omega}}
\nc{\rnc}{\renewcommand}

\nc{\catchset}{T}

\nc{\dg}{\dagger}
\nc{\dn}{\downarrow}

\nc{\1}{{\mathbbm{1}}}
\begin{document}

\title{One-shot rates for entanglement manipulation under non-entangling 
maps}

\author{Fernando~G.S.L.~Brand\~ao and Nilanjana Datta 
  \thanks{Fernando Brand\~ao (fgslbrandao@googlemail.com) is at the Physics Department of Universidade Federal de Minas Gerais, Brazil. 
Nilanjana Datta (N.Datta@statslab.cam.ac.uk) is in the Statistical Laboratory, Dept. of Applied Mathematics and Theoretical Physics, University of Cambridge. This work is part of the QIP-IRC supported by the European Community's Seventh Framework Programme
(FP7/2007-2013) under grant agreement number 213681, EPSRC (GR/S82176/0) as well as the 
Integrated Project Qubit Applications (QAP) supported by the IST directorate 
as Contract Number 015848'.  FB is supported by a "Conhecimento Novo" fellowship from Funda\c{c}\~ao 
de Amparo a Pesquisa do Estado de Minas Gerais (FAPEMIG). } 
}
\maketitle

\begin{abstract}

We obtain expressions for the optimal rates of one-shot entanglement manipulation under operations which generate a negligible
amount of entanglement. As the optimal rates for entanglement distillation and dilution in this paradigm, we obtain the max- and min-relative entropies of entanglement, the two logarithmic robustnesses of entanglement, and smoothed versions thereof. This gives a new operational meaning 
to these entanglement measures. Moreover, by considering the limit of many identical copies of the shared entangled state, we partially recover the recently found reversibility of
entanglement manipulation under the class of operations which asymptotically do not generate entanglement. 

\end{abstract}

\section{Introduction}

In the distant laboratory paradigm of quantum information theory, a system shared by two or more parties might have correlations that cannot be described by classical shared randomness; we say a state is entangled if it contains such intrinsically quantum correlations and hence cannot be created by local operations and classical communication (LOCC). Quantum teleportation \cite{BBC+93} shows that entanglement can actually be seen as a resource under the constraint that only LOCC operations are accessible. Indeed, one can use entanglement and LOCC to implement any operation allowed by quantum theory \cite{BBC+93}.
%, lifting the restriction that only LOCC operations could be implemented. 
The development of entanglement theory is thus centered in understanding, in a quantitative manner, the interconversion of one entangled state into another by LOCC, and their use for various information-theoretic tasks \cite{HHHH07, PV07}. 

In \cite{BBPS96}, Bennett \textit{et al} proved that entanglement manipulations of bipartite pure states, in the asymptotic limit of an
arbitrarily large number of copies of the state, are reversible. Given two 
bipartite pure states $\ket{\psi_{AB}}$ and $\ket{\phi_{AB}}$, 
the former can be converted into the latter by LOCC if, and only if, $E(\ket{\psi_{AB}})
\geq E(\ket{\phi_{AB}})$, where $E$ is the von Neumann entropy of 
either of the two reduced density matrices of the state. For mixed bipartite states, it turns out that the situation is rather more complex. For instance there are examples of mixed bipartite states, known as \textit{bound} entangled states \cite{HHH98}, which require a non-zero rate of pure state entanglement for their creation by LOCC in the limit of many copies, but from which no pure state entanglement can be extracted \cite{HHH98, VC01, YHHS05}. 

This inherent irreversibility in the asymptotic manipulation of entanglement led to the exploration of different scenarios for the study of entanglement, departing from the original one based on LOCC operations (see e.g. \cite{Rai97, Rai01, EVWW01, APE03, IP05, HOH02}). The main motivation in these studies was to develop a simplified theory of entanglement manipulation, with the hope that it would also lead to new insights into the physically motivated setting of LOCC manipulations. 

Recently one possible such scenario has been identified. In Refs. \cite{BP08a, BP, Bra} the manipulation of entanglement under any operation which generates a negligible amount of entanglement, in the limit of many copies, was put forward. Remarkably, it was found that one recovers for multipartite mixed states the reversibility encountered for bipartite pure states under LOCC. In such a setting, only one measure is meaningful: the regularized relative entropy of entanglement 
\cite{VP, VW01}; it completely specifies when a multipartite state can be converted into another by the accessible operations. This framework has also found interesting applications to the LOCC paradigm, such as a proof that the LOCC entanglement cost is strictly positive for every multipartite entangled state \cite{Bra, BP09a} (see \cite{Piani09} for a different proof), new insights into separability criteria \cite{BS09}, and impossibility results for reversible transformations of pure multipartite states \cite{Bra10}. 

In this paper we analyze entanglement conversion of general multipartite states under non-entangling and approximately non-entangling operations in the \textit{single copy} regime (see e.g. \cite{wolf2}, \cite{ligong}, \cite{buscemidatta}, \cite{mmnd}, \cite{wcr} for other studies of the single copy regime in classical and quantum information theory). We will identify the single copy cost and distillation functions under non-entangling maps with the two logarithmic robustnesses of entanglement \cite{VT, HN, Bra05} (one of them also referred to as the max-relative entropy of entanglement \cite{nd2}), and the min-relative entropy of entanglement \cite{nd2}, respectively. On one hand, our findings give operational interpretation to these entanglement measures. On the other hand, they give further insight into the reversibility attained in the asymptotic regime. Indeed, we will be able to prove reversibility, under catalytic entanglement manipulations, by taking the asymptotic limit in our finite copy formulae and using a certain generalization of quantum Stein's Lemma proved in Ref. \cite{BP09a} (which is also the main technical tool used in \cite{BP08a, BP, Bra}). We hence partially recover the results of \cite{BP08a, BP, Bra}, where reversibility was proved without the use of entanglement catalysis.  

The paper is organized as follows. In Section \ref{prelim} we introduce the 
necessary notation and definitions. Section \ref{main} contains our main
results, stated as Theorems \ref{smoothmin}-\ref{thmeq}. These theorems are then proved
in Sections \ref{min},\ref{lrob}, \ref{max} and \ref{eq}, respectively. 

\section{Notation and Definitions}
\label{prelim}
Let ${\cal B}(\sH)$ denote the algebra of linear operators acting on a
finite--dimensional Hilbert space $\sH$, and let ${\cal B}^{+}({\cal H})
\subset {\cal B}(\sH)$
denote the set of positive operators acting in $\sH$. Let  ${\cal D}(\sH)
\subset {\cal B}^{+}({\cal H})$
denote the set of states (positive operators of unit trace). 

Given a multipartite Hilbert space $\sH = \sH_1 \otimes ... \otimes \sH_m$, we say 
a state $\sigma \in {\cal D}(\sH_1 \otimes ... \otimes \sH_m)$ is separable if there 
are local states $\sigma_j^k \in {\cal D}(\sH_k)$ and a probability distribution 
$\{ p_j \}$ such that
\begin{equation}
\sigma = \sum_j p_j \sigma_j^1 \otimes ... \otimes  \sigma_j^m.
\end{equation}
We denote the set of separable states by ${\cal{S}}$. 

For given orthonormal bases $\{|i^A\rangle\}_{i=1}^d$ and
$\{|i^B\rangle\}_{i=1}^d$ in isomorphic Hilbert spaces ${\cal{H}}_A$
and ${\cal{H}}_B$ of dimension $d$, a maximally entangled state (MES)
of rank $M \le d$ is given by
$$ |\Psi_M^{AB}\rangle= \frac{1}{\sqrt{M}} \sum_{i=1}^M |i^A\rangle |i^B\rangle.$$
We define the fidelity of two quantum states $\rho$, $\sigma$ as
\begin{equation}
F(\rho, \sigma) = \left( \tr \sqrt{\sqrt{\sigma}\rho \sqrt{\sigma}} \right)^2.
\end{equation}
Finally, we denote the support of an operator $X$ by $\supp(X)$.
Throughout this paper we
restrict our considerations to finite-dimensional Hilbert spaces, and 
we take the logarithm to base $2$.

In \cite{nd1} two generalized relative entropy quantities, referred
to as the min- and max- relative entropies, were introduced. These are
defined as follows.
\begin{definition}
Let $\rho \in {\cal D}({\cal H})$ and $\sigma \in 
{\cal B}^{+}({\cal H})$ be such that $\supp(\rho) \subseteq \supp(\sigma)$. Their max-relative entropy is given by
  \be
    D_{\max}(\rho|| \sigma)
  :=
    \log \min \{ \lambda: \, \rho\leq \lambda \sigma \},
    \ee
while their min-relative entropy is given by
  \be
    D_{\min}(\rho|| \sigma)
  :=    - \log \tr\bigl(\Pi_{\rho}\sigma\bigr) \ ,
  \label{dmin}
\ee
where $\Pi_{\rho}$ denotes the projector onto $\supp(\rho)$ \footnote{Note that $D_{\min}(\rho|| \sigma)$ is well-defined whenever $\text{supp}(\rho) \cap \text{supp}(\sigma)$ is not empty.}.
\end{definition}
As noted in \cite{nd1, mmnd}, $D_{\min}(\rho|| \sigma)$ is the relative R\'enyi 
entropy of order $0$. 

In \cite{nd2} two entanglement measures were defined in terms of the above quantities. 
\begin{definition}
The max-relative entropy of entanglement of $\rho \in {\cal D}({\cal H})$ is given by 
\be E_{\max}(\rho):= \min_{\sigma \in {\cal{S}}} D_{\max} (\rho||\sigma),
\label{entmeasure}
\ee
while its min-relative entropy of entanglement is given by
\be 
E_{\min}(\rho):= \min_{\sigma \in {\cal{S}}} D_{\min} (\rho||\sigma),
\label{entmin}
\ee
\end{definition}

It turns out \cite{nd2} that $E_{\max}(\rho)$ is not really a new quantity, but is actually equal to the logarithmic version of the global robustness of 
entanglement, given by \cite{Bra05}
\be
LR_G(\rho) := \log (1 + R_G(\rho)),
\label{glr}
\ee
where $R_G(\rho)$ is the global robustness of entanglement \cite{VT, HN} 
defined as
\be
R_G(\rho) := \min_{s \in \mathbb{R}} \Bigl(s \ge 0: \exists \,\omega \in {\cal{D}} \,\,{\rm{s.t.}}\,\, 
\frac{1}{1+s}\rho + \frac{s}{1+s}\omega \in  {\cal{S}} \Bigr).
\label{gr}
\ee
 
Another quantity of relevance in this paper is the robustness of entanglement \cite{VT}, 
denoted by $R(\rho)$. Its definition is analogous to that of $R_G(\rho)$ 
except that the states $\omega$ in Eq.\reff{gr} are restricted to separable 
states. Its logarithmic version is defined as follows.
\begin{definition}
The logarithmic robustness of entanglement of $\rho \in {\cal D}({\cal H})$ is given by
\begin{equation}
LR(\rho) := \log(1 + R(\rho)). 
\end{equation}
\end{definition}

We also define smoothed versions of the quantities we consider as follows (see also \cite{BP09a, marco}).
\begin{definition}
For any $\eps >0$, the smooth max-relative entropy of 
entanglement of $\rho \in {\cal D}({\cal H})$
is given by
\be
{E}_{\max}^\eps (\rho) := \min_{\bar{\rho} \in B^{\eps}(\rho)} E_{\max}(\bar{\rho}),
\ee
where $B^{\eps}(\rho) := \{\bar \rho \in {\cal D}({\cal H}) : F(\bar \rho, \rho) \geq 1 - \eps  \}$. 

The smooth logarithmic robustness of entanglement of $\rho \in {\cal D}({\cal H})$ in turn is given by
\be\label{eps_LR}
LR^\eps (\rho) := \min_{\bar{\rho} \in B^{\eps}(\rho)} LR(\bar{\rho}).
\ee

Finally, the smooth min-relative entropy of entanglement of $\rho \in {\cal D}({\cal H})$ is defined as
\be
{E}_{\min}^\eps (\rho) := \max_{0 \leq A \leq \id \atop{\tr(A \rho) \geq 1 - \eps}} \min_{\sigma \in {\cal S}} \left(- \log \tr(A \sigma) \right).
\ee
\end{definition}

We note that the definition of ${E}_{\min}^\eps(\rho)$ which we use in this paper is different from the one introduced in \cite{nd2}, where the smoothing was performed over an $\eps$-ball around the state $\rho$, in analogy with the smooth version of ${E}_{\max}^\eps(\rho)$ given above. 
%The justification of this new definition arises from the fact that the 
%expression for the min-relative entropy of entanglement of a state $\rho$ 
%involves the projection onto its support,
%but the support of an operator in an $\eps$-ball around $\rho$ may be very different from that of the support of $\rho$ itself. So, in defining ${E}_{\min}^\eps (\rho)$, % instead of replacing 
%$\Pi_\rho$ by the projection onto the support of an operator in the $\eps$-ball around $\rho$, it is meaningful to replace $\Pi_\rho$ itself by an operator 
%$0\le A\le I$, which almost projects onto the support of $\rho$. 
Note also that while this new smoothing is a priori inequivalent to the one
in \cite{nd2}, it is equivalent to the ``operator-smoothing'' introduced in \cite{buscemidatta}, which, in addition, gives rise to a continuous family of smoothed relative R\'enyi entropies.

We will consider regularized versions of the smooth min- and max-relative entropies
of entanglement
\bea
{\E}_{\min}^\eps (\rho)&:=& \liminf_{n\rightarrow \infty}\frac{1}{n} E_{\min}^\eps(\rho^{\otimes n}),
\nonumber\\
{\E}_{\max}^\eps (\rho)&:=& \limsup_{n\rightarrow \infty}\frac{1}{n} E_{\max}^\eps(\rho^{\otimes n}),\nonumber\\
\eea
and the quantities
\bea
{\E}_{\min} (\rho) &:=& \lim_{\eps \rightarrow 0} {\E}_{\min}^\eps (\rho)
\nonumber\\
{\E}_{\max} (\rho) &:=& \lim_{\eps \rightarrow 0} {\E}_{\max}^\eps (\rho)
\label{deff}
\eea
In \cite{BP09a, nd2} it was proved that ${\E}_{\max} (\rho)$ is equal to the
regularized relative entropy of entanglement \cite{VP, VW01}
\be
E_R^\infty(\rho):= \lim_{n\rightarrow \infty} \frac{1}{n}E_R(\rho^{\otimes n}),
\label{rel11}
\ee
where
\begin{equation}
E_{R}(\omega) := \min_{\sigma \in {\cal S}} S(\omega || \sigma),
\end{equation}
is the relative entropy of entanglement and $S(\omega || \sigma) := \tr(\rho(\log(\rho) - \log(\sigma)))$ the quantum relative entropy.

In this paper we prove that also ${\E}_{\min} (\rho)$ is equal to $E_R^\infty(\rho)$ (see Theorem 4).

We can now be more precise about the classes of maps we consider for the manipulation of entanglement, introduced in \cite{BP08a, BP}. 

\begin{definition}
A completely positive trace-preserving (CPTP) map $\Lambda$ is said to be a non-entangling (or separability preserving) map 
if $\Lambda(\sigma)$ is separable for any separable state $\sigma$. We denote the class of such maps by SEPP \footnote{The acronym comes from the name {\em{separability preserving.}}}.
\end{definition}

\begin{definition}
For any given $\delta >0$ we say a map $\Lambda$ is a $\delta$-non-entangling map if $R_G(\Lambda(\sigma)) \le \delta$ for every separable state $\sigma$. We denote the class of such maps by $\delta$-SEPP.
\end{definition}

In the following sections we will consider entanglement manipulations under non-entangling and $\delta$-non-entangling maps. We first give the definitions of achievable and optimal rates of entanglement manipulation protocols under a general class of maps, in order to make the subsequent discussion more transparent. In the definitions we will consider maps from a multipartite state to a maximally entangled state and vice-versa. It should be understood that the first two parties share the maximally entangled state, while the quantum state of the other parties is trivial (one-dimensional). 

\begin{definition} 
The one-shot entanglement cost of $\rho$ under the class of operations $\Theta$ is defined as
\begin{eqnarray}
& &E_{C,\Theta}^{(1), \eps}(\rho)\\
& & := \min_{M, \Lambda} \{ \log M  :F( \rho, \Lambda(\Psi_M))\geq 1 - \eps, \Lambda \in \Theta, M \in \mathbb{Z}^+ \}. \nonumber
\end{eqnarray}
\end{definition}

We also consider a \textit{catalytic} version of entanglement dilution under $\delta$-non-entangling maps. 

\begin{definition}
The one-shot catalytic entanglement cost of $\rho$ under a class of quantum 
operations ${\Theta}$ is defined as 
\begin{eqnarray}
\tilde E_{C,\Theta}^{(1), \eps}(\rho) := \min_{M, K, \Lambda} && \{ \log M  : \Lambda(\Psi_M \otimes \Psi_K) = \rho' \otimes \Psi_K, \nonumber \\ && F( \rho, \rho')\geq 1 - \eps, \Lambda \in \Theta, M, K \in \mathbb{Z}^+ \}. \nonumber
\end{eqnarray}
\end{definition}
 
Finally, the next definition formalizes the notion of single-shot entanglement distillation under general 
classes of maps.    
 
\begin{definition} 
The one-shot distillable entanglement of $\rho$ under a class of quantum 
operations ${\Theta}$ is defined as 
\begin{eqnarray}
& &E_{D,\Theta}^{(1), \eps}(\rho)\\
& & := \max_{M, \Lambda} \{ \log M  :F( \Lambda(\rho), \Psi_M)\geq 1 - \eps, \Lambda \in \Theta, M \in \mathbb{Z}^+ \}. \nonumber
\end{eqnarray}
\end{definition}
 
In the following we shall consider ${\Theta}$ to be either the class of SEPP maps or the
class of $\delta$-SEPP maps for a given $\delta >0$.

\section{Main Results} \label{main}

The main results of the paper are given by the following four theorems.
They provide operational interpretations of the smooth max- and min-relative
entropies of entanglement, and the logarithmic version of the robustness of entanglement,
in terms of optimal rates of one-shot entanglement manipulation
protocols.

The first theorem relates the smoothed min-relative entropy of entanglement to the single-shot distillable entanglement under 
non-entangling maps.
\begin{theorem}
\label{smoothmin}
For any state $\rho$ and any $\varepsilon \geq 0$,
\be
\lfloor E_{\min}^\varepsilon (\rho) \rfloor \le E_{D, SEPP}^{(1), \eps}(\rho)
\le  E_{\min}^{\varepsilon}(\rho).
\ee 
\end{theorem}

The following theorem relates the smoothed logarithmic robustness of entanglement to the one-shot entanglement cost under 
non-entangling maps. 
\begin{theorem}
\label{logrob}
For any state $\rho$ and any $\varepsilon \geq 0$,
\be
LR^{\varepsilon}(\rho) \le E_{C,{\rm{SEPP}}}^{(1),\varepsilon }(\rho) \le LR^{\varepsilon}(\rho) + 1.
\label{47}
\ee
\end{theorem}

We also prove an analogous theorem to the previous one, but now relating the logarithmic global robustness (alias max-relative 
entropy of entanglement) to the one-shot catalytic entanglement cost under $\delta$-non-entangling maps. 
\begin{theorem}
\label{emax}
For any $\delta, \eps > 0$ there exists a positive integer $K$, such that for any state $\rho$
\bea
E_{\max}^{\eps}(\rho \otimes \Psi_K) &-& \log K  - \log(1+\delta)
\le {\widetilde{E}}_{C, \delta-SEPP}^{(1), \eps}(\rho) \nonumber\\
&\le& E_{\max}^{\eps}(\rho \otimes \Psi_K) - \log(1 - \eps) - \log K + 1.\nonumber \\
\eea
We can take in particular $K = \lceil 1 + \delta^{-1} \rceil$.
\end{theorem}

Finally we show that we can partially recover the reversibility of entanglement manipulations under asymptotically non-entangling maps \cite{BP08a, Bra05} from the results derived in this paper and the quantum hypothesis testing result of \cite{BP09a}. 
\begin{theorem}
\label{thmeq}
For every state $\rho \in {\cal D}({\cal H})$, 
\begin{equation}
{\E}_{\min} (\rho) = {\E}_{\max} (\rho) = E_R^{\infty}(\rho). 
\end{equation}
\end{theorem}
From Theorems \ref{smoothmin} and  \ref{emax} we then find that the distillable entanglement and the \textit{catalytic} entanglement cost under asymptotically non-entangling maps are the same. In Refs. \cite{BP08a, Bra05} one could show the same result without the need of catalysis. Here we need the extra resource of catalytic maximally entangled states because we want to ensure that already on a single-copy level, our operations only generate a negligible amont of entanglement; in Refs. \cite{BP08a, Bra05}, in turn, this is only the case for a large number of copies of the state. 

In more detail: we define the distillable entanglement under non-entangling operations as 
\begin{equation}
E_D^{ne}(\rho) := \lim_{\eps \rightarrow 0} \lim_{n \rightarrow \infty} \frac{1}{n} E_{D, SEPP}^{(1), \eps}(\rho^{\otimes n}).
\end{equation}
It then follows easily from Theorem \ref{smoothmin} and Theorem \ref{thmeq} that $E_D^{ne}(\rho) = E_R^{\infty}(\rho)$. 

The catalytic entanglement cost under asymptotic non-entangling operations, in turn, is defined as 
\begin{equation}
E_C^{ane}(\rho) := \lim_{\eps \rightarrow 0} \lim_{\delta \rightarrow 0} \lim_{n \rightarrow \infty} \frac{1}{n} {\widetilde{E}}_{C, \delta-SEPP}^{(1), \eps}(\rho).
\end{equation}
That $E_C^{ane}(\rho) = E_R^{\infty}(\rho)$ then follows from Theorems \ref{emax} and \ref{thmeq}.

We note that it was already proven in Refs.\cite{nd2, BP09a} that ${\E}_{\max} (\rho) = E_R^{\infty}(\rho)$. Our contribution is to show that also the regularization of the smooth min-relative entropy of entanglement is equal to the regularized relative entropy of entanglement.

\section{Proof of Theorem~\ref{smoothmin}}
\label{min}
The proof of Theorem~\ref{smoothmin} will employ the following lemma.
\begin{lemma} \label{monotonicityEmin}
For any $\Lambda \in {\rm{SEPP}}$, 
\begin{equation}
E_{\min}^{\eps}(\rho) \geq E_{\min}^{\eps}(\Lambda(\rho))
\end{equation}
\end{lemma}
\begin{proof}
Let $0 \leq A \leq \id$ be such that $\tr(A \Lambda(\rho)) \geq 1 - \eps$ and $E_{\min}^{\eps}(\Lambda(\rho)) = \min_{\sigma \in {\cal S}} (- \log \tr(A \sigma))$. Setting $\sigma_{\rho}$ as the optimal state in the definition of $E_{\min}^{\eps}(\rho)$,
\bea
E_{\min}^{\eps}(\rho) &\geq& - \log \tr(\Lambda^{\cal y}(A) \sigma_{\rho}) \nonumber \\
&=& - \log \tr(A \Lambda(\sigma_{\rho})) \nonumber \\
&\geq& \min_{\sigma \in {\cal S}} \left( - \log \tr(A \sigma) \right) \nonumber \\
&=& E_{\min}^{\eps}(\Lambda(\rho)).
\eea
where $\Lambda^{\cal y}$ is the adjoint map of $\Lambda$. In the first line we used that $0 \leq \Lambda^{\cal y}(A) \leq \id$ and $\tr(\Lambda^{\cal y}(A)\rho) = \tr(A \Lambda(\rho))\geq 1 - \eps$, while in the third line we use the fact that $\Lambda(\sigma_\rho)$ is separable, since 
$\Lambda \in {\rm{SEPP}}$.
\end{proof}

\begin{proof}[Theorem~\ref{smoothmin}]
We first prove that $E_{D,{\rm{SEPP}}}^{(1), \eps}\ge \lfloor E_{\min}^{\eps}(\rho) \rfloor.$ For this it suffices to prove that 
any $R \le  \lfloor E_{\min}^{\eps}(\rho) \rfloor$ is an achievable
one-shot distillation rate for $\rho$.

Consider the class of completely positive trace-preserving 
maps $\Lambda \equiv \Lambda_A$ (for an operator $0\le A \le I$) 
whose action on a state $\rho$ is given as follows:
\be
\Lambda(\rho) := \tr(A\rho) \Psi_M + \tr\bigl((I-A)\rho\bigr) \frac{(I - \Psi_M)}{M^2 - 1},
\label{vintecinco}
\ee
for any state $\rho\in {\cal{D}}({\cal{H}})$. An isotropic state $\omega$, as the one appearing on the right-hand side of Eq. (\ref{vintecinco}), is separable if and only if 
$\tr(\omega \Psi_M) \le 1/M$ \cite{horo1}. Hence, the map $\Lambda$ is SEPP if,
and only if, for
any separable state $\sigma$,
$\tr(\Lambda(\sigma) \Psi_M) \le 1/M$, or equivalently,
\be
\tr(A \sigma) \le \frac{1}{M}.
\label{e11}
\ee

We now choose $A$ as the optimal POVM element in the definition of $E_{\min}^{\eps}(\rho)$ and set $M = 2^{\lfloor  E_{\min}^{\eps}(\rho) \rfloor}$. 

On one hand, as $\tr(A \rho) \geq 1 - \eps$, we find that $F(\Lambda(\rho), \Psi_M) \geq 1 - \eps$. On the other hand, by the definition of $E_{\min}^{\eps}(\rho)$, we have that 
\be
2^{- E_{\min}^{\eps}(\rho)} = \max_{\sigma \in {\cal S}} \tr(A \sigma)
\ee
and hence $\tr(A \sigma) \leq 1/M$ for every separable state $\sigma$, which 
ensures that the map $\Lambda$ defined by \reff{vintecinco} is a SEPP map. Hence,
$\log M = \lfloor E_{\min}^\eps (\rho)\rfloor$ is an achievable rate and
$E_{D,{\rm{SEPP}}}^{(1), \eps}\ge \lfloor E_{\min}^{\eps}(\rho) \rfloor.$

We next prove the converse, namely that $E_{D,{\rm{SEPP}}}^{(1), \eps}(\rho) \le  E_{\min}^{\eps}(\rho).$
Suppose $\Lambda$ is the optimal SEPP map such that $F(\Lambda(\rho), \Psi_M) \geq 1 - \eps$, with $\log M = E_{D,\eps}^{(1)}(\rho)$.

By Lemma \ref{monotonicityEmin} we have
\bea
E_{\min}^{\eps}(\rho) &\geq& E_{\min}^{\eps}(\Lambda(\rho)) \nonumber \\ 
&=& \max_{0 \leq A \leq \id \atop{\tr(A \Lambda(\rho)) \geq 1 - \eps}} \min_{\sigma \in {\cal S}} \left(- \log \tr(A \sigma) \right) \nonumber \\
&\geq& \min_{\sigma \in {\cal S}} \left(- \log \tr(\Psi_M \sigma) \right) \nonumber \\
&=& \log M \nonumber \\
&=& E_{D,\eps}^{(1)}(\rho),
\eea
where we used that $0 \leq \Psi_M \leq \id$ and $\tr(\Lambda(\rho)\Psi_M) \geq 1 - \eps$ and that $\tr(\Psi_M \sigma) \leq 1/M$ for every separable state $\sigma$.
\end{proof}

\section{Proof of Theorem~\ref{logrob}}
\label{lrob}
\begin{proof}
To prove the upper bound in \reff{47}, consider the 
quantum operation $\Lambda$ acting on a state $\omega$ as follows:
\be 
\Lambda(\omega) = \tr(\Psi_M \omega)\rho_{\varepsilon} + \bigl[1 -  \tr(\Psi_M \omega)\bigr] 
\pi,
\label{57}
\ee
where {{$\rho_\eps$ is the state in $B^\eps(\rho)$ which achieves the 
minimum in the definition \reff{eps_LR}of the smooth logarithmic robustness}},
and $\pi$ is a separable state such that the state
$$\sigma:= \bigl( \rho_{\varepsilon} + (M-1) \pi\bigr)/M,$$
is separable for the choice $M = 1 + \lceil R(\rho_{\varepsilon}) \rceil$.
%the optimal state for the robustness $R(\rho_{\varepsilon})$ of $\rho_{\varepsilon}$, 
%which is such that $LR^{\varepsilon}(\rho) = LR(\rho_{\varepsilon})$. 

We can rewrite Eq. \reff{57} as
\be 
\Lambda(\omega) = q \bigl[\frac{\rho_{\varepsilon} + (M-1) \pi}{M}\bigr]
+ (1-q) \pi,
\ee
where $q = M\tr(\Psi_M \omega)$. For a separable state 
$\omega$, $\tr(\Psi_M \omega) \le 1/M$ \cite{horo_sep}, and
hence $0\le q \le 1$.
By the convexity of the robustness \cite{rob_convex} we have that, for any
separable state $\omega$,
$$ R(\Lambda(\omega)) \le q R(\sigma) + (1-q) R(\pi).$$
Note that $R(\pi)=0$
since $\pi$ is separable. Moreover, since $R(\sigma)=0$ for $M = 1 + \lceil R(\rho_{\varepsilon}) \rceil$, we have
$R(\Lambda(\omega))=0$, ensuring that
the map $\Lambda$ is non-entangling. 
 
Note that $\Lambda(\Psi_M) = \rho_{\varepsilon}$, with the corresponding rate of 
$\log M = \log (1 +\lceil R(\rho_{\varepsilon}) \rceil) \leq LR^{\varepsilon}(\rho) + 1$.
This then yields the upper bound in Theorem \ref{logrob}.

To prove the lower bound in \reff{47}, let $\Lambda $ denote a SEPP
map yielding entanglement dilution with a fidelity of at least $1-\varepsilon$, for a state $\rho$, i.e.
$\Lambda_M(\Psi_M) = \rho_{\varepsilon}$, with $F(\rho, \rho_{\varepsilon}) \geq 1 - \varepsilon$, and $\log M = E_{C, {\rm{SEPP}}}^{(1), \varepsilon}$.
The monotonicity of log robustness under SEPP maps \cite{BP} yields
\bea
LR^{\varepsilon}(\rho) \leq LR(\rho_{\varepsilon}) &=& LR(\Lambda(\Psi_M))\nonumber\\
 & \le& LR(\Psi_M) \nonumber\\
&=& \log M =E_{C, {\rm{SEPP}}}^{(1), \varepsilon}.
\nonumber\\
\eea
\end{proof}

\section{Proof of Theorem~\ref{emax}}
\label{max}
The following lemmata will be employed in the proof of Theorem~\ref{emax}
\begin{lemma} \label{monoemax}
For any $\delta >0$ and $\Lambda \in \delta$-{\rm{SEPP}}, 
\begin{equation}
E_{\max}^{\eps}(\rho) \geq E_{\max}^{\eps}(\Lambda(\rho)) - \log(1+ \delta)
\end{equation}
\end{lemma}
\begin{proof}

Let $\rho_\eps$ be the optimal state in the definition of $E_{\max}^{\eps}(\rho)$, i.e., $E_{\max}^{\eps}(\rho) = E_{\max}(\rho_{\eps})$. By the 
monotonicity of the fidelity under CPTP maps we have that
$F(\Lambda(\rho), \Lambda(\rho_\eps))\ge F(\rho, \rho_\eps) \ge 1- \eps.$
Hence, using Lemma IV.1 of \cite{BP}
\bea
E_{\max}^{\eps}(\Lambda(\rho)) &\le & E_{\max}(\Lambda(\rho_\eps))\nonumber\\
&\le & E_{\max}(\rho_\eps) + \log (1 + \delta)\nonumber\\
&=& E_{\max}^{\eps}(\rho) + \log (1 + \delta).
\eea
\end{proof}

\begin{lemma} \label{newlemmads}
For every $\rho \in {\cal D}({\cal H})$ and $\varepsilon > 0$, there is a state $\mu_{\varepsilon}$ of the form 
\begin{equation}\label{mu}
\mu_{\varepsilon} := (1 - \lambda) \rho_{\varepsilon} \otimes \Psi_K + \lambda \theta \otimes \left( \frac{\id - \Psi_K}{K^2 - 1} \right),
\end{equation}
with $K \in  \mathbb{Z}^+ \}$, $\theta, \rho_{\varepsilon} \in {\cal D}({\cal H})$, $F(\rho, \rho_\eps) \ge 1 - \eps$, and $\lambda \leq \varepsilon$, such that
\begin{equation}
E_{\max}^{\varepsilon}(\rho \otimes \Psi_K) \ge E_{\max}(\mu_{\eps}).
\end{equation}
\end{lemma}

\begin{proof}
Let $\mu_{\varepsilon}'$ be such that $E_{\max}^{\varepsilon}(\rho \otimes \Psi_K) = E_{\max}(\mu_{\varepsilon}')$. Then there is a separable state $\sigma$ such that
\begin{equation} \label{newlemmaeq1}
\mu_{\varepsilon}' \leq 2^{E_{\max}^{\varepsilon}(\rho \otimes \Psi_K)}\sigma
\end{equation}
and $F(\mu_{\varepsilon}', \rho \otimes \Psi_K) \geq 1 - \varepsilon$. 
Consider the twirling map 
\begin{equation}
\Delta(X) := \int_{Haar} dU (U \otimes U^*) X (U \otimes U^*)^{\cal y}
\end{equation}
and define $\mu_{\eps} := (\id \otimes \Delta)(\mu_{\varepsilon}')$. Then, because $\Delta$ is entanglement breaking \cite{ruskai_shor} we can write 
\begin{equation}\label{mu_e}
\mu_{\varepsilon} := (1 - \lambda) \rho_{\varepsilon} \otimes \Psi_K + \lambda \theta \otimes \left( \frac{\id - \Psi_K}{K^2 - 1} \right),
\end{equation}
for $\theta, \rho_{\varepsilon} \in {\cal D}({\cal H})$ and $0 \leq \lambda \leq 1$. From Eq. (\ref{newlemmaeq1}), 
\begin{equation}
\mu_{\varepsilon} \leq 2^{E_{\max}^{\varepsilon}(\rho \otimes \Psi_K)} (\id \otimes \Delta)\sigma.
\end{equation}
Since $\Delta$ is LOCC, $(\id \otimes \Delta)\sigma$ is separable and we get $E_{\max}(\mu_{\varepsilon}) \leq E_{\max}^{\varepsilon}(\rho \otimes \Psi_K)$. Moreover, from the monotonicity of the fidelity under CPTP maps, $F(\mu_{\varepsilon}, \rho \otimes \Psi_K) \geq 1 - \varepsilon$. From this and (\ref{mu}) it follows that
$$(1-\lambda) \ge F(\rho, \rho_\eps) \ge 1 - \eps,$$ and thus, $\lambda \leq \varepsilon$. 
\end{proof}

\begin{proof}[Theorem~\ref{emax}]
Let us start by proving the achievability part, namely that for every $\delta > 0$ we can find a positive integer $K$ such that ${\widetilde{E}}_{C, \delta-{\rm{SEPP}}}^{(1), \eps}(\rho) \le E_{\max}^{\eps}(\rho \otimes \Psi_K) - \log(1 - \eps) - \log K$. 

From Lemma \ref{newlemmads} we know there is a state $\rho_{\eps}$ such that $F(\rho_{\eps}, \rho) \geq 1 - \eps$ and $E_{\max}(\rho_{\eps} \otimes \Psi_K) \leq E_{\max}^{\eps}(\rho \otimes \Psi_K) - \log(1 - \eps)$. This
can be seen as follows: {{Let $\mu_\eps$ be a state of the form given 
by \reff{mu}}}. From the definition of the max-relative entropy of entanglement (Definition \ref{entmeasure}) it follows that
\begin{eqnarray}
\mu_\eps &\le & 2^{E_{\max}(\mu_\eps)} \sigma',\nonumber\\
&\le & 2^{E_{\max}^{\varepsilon}(\rho \otimes \Psi_K)}\sigma'.
\eea 
for some separable state $\sigma' \in {\cal{B}}({\cal{H}})$, where we get the second inequality
by using Lemma \ref{newlemmads}. Substituting the expression \reff{mu} of $\mu_\eps$ we get 
\bea
&&(1 - \lambda) \rho_{\varepsilon} \otimes \Psi_K + \lambda \theta \otimes \left( \frac{\id - \Psi_K}{K^2 - 1} \right)\nonumber\\
&\le & 2^{E_{\max}^{\varepsilon}(\rho \otimes \Psi_K)}\sigma'.
\eea
This yields,
\be 
(1 - \lambda) \rho_{\varepsilon} \otimes \Psi_K \le 2^{E_{\max}^{\varepsilon}(\rho \otimes \Psi_K)}\sigma',
\ee
and hence, 
$$\rho_{\varepsilon} \otimes \Psi_K\le 2^{E_{\max}^{\varepsilon}(\rho \otimes \Psi_K)}2^{-\log(1-\lambda)}
\sigma',$$
which in turn implies that 
$$\rho_{\varepsilon} \otimes \Psi_K\le 2^{E_{\max}^{\varepsilon}(\rho \otimes \Psi_K) - \log(1 - \eps)}
\sigma',$$
since $\lambda \le \eps$. Therefore, for $K = \lceil 1 + \delta^{-1}  \rceil$ and $M = \lceil K^{-1} 2^{ E_{\max}^{\eps}(\rho \otimes \Psi_K) - \log(1 - \eps)} \rceil$, we can always find a state $\pi$ such that $\bigl((\rho_\eps \otimes \Psi_K) + (MK-1) \pi\bigr)$ is an unnormalized separable state.

Define the map
\begin{eqnarray}
\Lambda(\omega) &=& \bigl[  \tr((\Psi_M \otimes \Psi_K)  \omega) \bigr]
\bigl(\rho_{\eps} \otimes \Psi_{K}) \nonumber \\ &+&
\bigl[\tr((\id - \Psi_M \otimes \Psi_K) \omega)\bigr]\pi,
\label{qop2}
\end{eqnarray}

We now show that with our choice of parameters the map $\Lambda$ is $\delta$-SEPP. 
First note that since for any
separable state $\sigma  \in {\cal{B}}({\cal{H}} \otimes {\cal{H}})$ 
$$\tr\bigl( (\Psi_M \otimes \Psi_K) \sigma\bigr) \le \frac{1}{MK},$$
we can write 
\be
\Lambda(\sigma) = p (\rho_\eps \otimes \Psi_K) + (1-p) \pi,
\label{15}
\ee
where $p \le \frac{1}{MK}$. This in turn can be written as
\be 
\Lambda(\sigma) = q\bigl[\frac{ (\rho_\eps \otimes \Psi_K) + (MK-1)) \pi}{MK}\bigr]
+ (1-q)\pi,
\label{e27}
\ee
where $q = pMK$. Since $0\le p \le 1/MK$, we have that $0\le q \le 1$.
Note that the first term in parenthesis in \reff{e27} 
is separable, due to the choice of $\pi$. Using the
convexity of the global robustness we then conclude that 
$R_G(\Lambda(\sigma)) \le R_G(\pi)$, for any separable 
state $\sigma$.

Further, from the choice of $M$ and $K$ it follows that
$$R_G(\pi) \le \frac{1}{R_G(\rho_\eps \otimes \Psi_K) } \le \frac{1}{K-1}\le \delta.$$
The first inequality follows from the fact that if $(\rho + s \sigma)$ is an unnormalized separable 
state, then so is $(\sigma + (1/s) \rho)$, and by noting that
$$\frac{\rho+s\sigma}{1+s} = \frac{\sigma + s^{-1}\rho}{1 + s^{-1}}.$$
The second inequality follows from the monotonicity of $R_G$ under LOCC \cite{VT}, which implies $R_G(\rho_{\eps}\otimes \Psi_K) \geq R_G(\Psi_K)$ and the fact $R_G(\Psi_K) = K - 1$ \cite{VT}. Finally, the third is a consequence of the choice of $K$. 

Note that for $\omega = \Psi_M\otimes \Psi_K$,
\be\tlambda(\tomega) = \Lambda (\Psi_M\otimes \Psi_K )
= \rho_\eps\otimes \Psi_{K}.
\label{two}
\ee
Hence the protocol yields a state $\rho_\eps$ with $F(\rho, \rho_\eps) \ge 1 - \eps$ and the additional maximally entangled state $\Psi_{K}$ which was employed in the start of the protocol. Its role in the protocol is to ensure that the quantum operation $\tlambda$ is a $\delta$-SEPP 
map for any given $\delta >0$.
Since the maximally entangled states $\Psi_M$ and $\Psi_{K}$ were employed in the protocol
and $\Psi_{K}$ was retrieved unchanged, the rate
$
R = (\log M + \log M') - \log M'= \log M \leq E^\eps_{\max}(\rho \otimes \Psi_K) - \log K  - \log(1 - \eps) + 1,$
is achievable. 

Next we prove the bound ${\widetilde{E}}_{C, \delta-{\rm{SEPP}}}^{(1), 0} \ge E_{\max}^\eps(\rho) - \log K -\log (1+\delta)$. 
Let 
$\tlambda$ be a $\delta$-SEPP map for which
$$\Lambda(\Psi_M \otimes \Psi_K) = \rho_\eps \otimes \Psi_{K}.
$$
with ${\widetilde{E}}_{c, \delta-{\rm{SEPP}}}^{(1), \eps}= \log M$.

Then by Lemma \ref{monoemax},
\bea
E_{\max}^\eps(\rho \otimes \Psi_K) & \le & E_{\max}(\rho_\eps \otimes \Psi_K) \nonumber\\
&=&  E_{\max}(\Lambda(\Psi_M \otimes \Psi_K)) \nonumber\\
&\le & E_{\max}(\Psi_M \otimes \Psi_K) + \log (1+\delta) \nonumber\\
&=& \log M + \log K + \log (1+\delta).
\eea
Hence 
\be
\log M \ge  E_{\max}^\eps(\rho \otimes \Psi_K) - \log K - \log(1 + \delta).
\ee
\end{proof}

\section{Equivalence with the regularized relative entropy of entanglement}
\label{eq}
In this section we prove Theorem \ref{thmeq}. The main ingredient in the proof is a certain generalizaton of Quantum Stein's Lemma proved in Refs. \cite{Bra, BP09a} and stated below as 
Lemma \ref{BP09} for the special case of the separable states set.

\begin{lemma} \label{BP09}
Let $\rho \in {\cal D}({\cal H})$. Then 

(\textit{Direct part}): For every  $\eps > 0$ there exists a sequence of POVMs $\{ A_n, \id - 
A_n \}_{n \in \mathbb{N}}$ such that
\begin{equation}
        \lim_{n \rightarrow \infty} \tr((\id - A_n) \rho^{\otimes n}) = 0 
\end{equation}
and for every $n \in \mathbb{N}$ and $\omega_n \in {\cal S}({\cal H}^{\otimes n})$,
\begin{equation} \label{exponentialdecay}
        - \frac{\log \tr(A_n \omega_n)}{n} + {\eps} \geq 
        E_{\cal M}^{\infty}(\rho) .
\end{equation}

(\textit{Strong Converse}): If a real number $\eps > 0$ and a sequence of POVMs 
$\{ A_n, \id - A_n \}_{n \in \mathbb{N}}$ are such that for every $n \in \mathbb{N}$ and $\omega_n \in {\cal S}({\cal H}^{\otimes n})$,
\begin{equation}
         - \frac{\log( \tr(A_n \omega_n))}{n} - \eps \geq 
         E_{\cal M}^{\infty}(\rho),
\end{equation}
then
\begin{equation}
        \lim_{n \rightarrow \infty} \tr((\id - A_n) \rho^{\otimes n}) = 1. 
\end{equation}
\end{lemma}

\begin{proof}
(Theorem 4). In Refs. \cite{Bra, BP09a, nd2} it was established that
\begin{equation} \label{maxeq}
{\cal E}_{\max}(\rho) = E_R^{\infty}(\rho).
\end{equation}
We hence focus in showing that ${\cal E}_{\min}(\rho) \geq E_R^{\infty}(\rho)$, since ${\cal E}_{\min}(\rho) \leq E_R^{\infty}(\rho)$ follows from Eq. (\ref{maxeq}) and the fact that ${\cal E}_{\max}(\rho) \geq {\cal E}_{\min}(\rho)$ (which in turn is a direct consequence of their definitions). Let $\eps > 0$ and $\{ A_n \}$ be an optimal sequence of POVMs in the direct part of Lemma \ref{BP09}. Then for sufficiently large $n$, $\tr(\rho^{\otimes n}A_n) \geq 1 - \eps$ and thus
\begin{eqnarray}
{E}_{\min}^\eps (\rho^{\otimes n}) \geq  \min_{\sigma \in {\cal S}({\cal H}^{\otimes n})} \left(- \log \tr(A_n \sigma) \right) \geq n (E_R^{\infty}(\rho) - \eps),
\end{eqnarray}
where the last inequality follows from Eq. (\ref{exponentialdecay}). Dividing both sides by $n$ and taking the limit $n \rightarrow \infty$ we get
\begin{equation}
{\cal E}_{\min}^\eps (\rho) \geq  E_R^{\infty}(\rho) - \eps.
\end{equation}
Since this equation holds for every $\eps > 0$, we can finally take the limit $\eps \rightarrow 0$ to find
\begin{equation}
{\cal E}_{\min}(\rho) \geq E_R^{\infty}(\rho).
\end{equation}
\end{proof}

\section*{Acknowledgments}

The authors would like to thank Martin Plenio for many interesting discussions on the theme of this paper. 
This work is part of the QIP-IRC supported by the European Community's Seventh Framework Programme
(FP7/2007-2013) under grant agreement number 213681, EPSRC (GR/S82176/0) as well as the 
Integrated Project Qubit Applications (QAP) supported by the IST directorate 
as Contract Number 015848'. FB is supported by an EPSRC Postdoctoral Fellowship for Theoretical Physics.

\end{document}